\documentclass[10pt]{article}
\usepackage{fancyhdr}
\usepackage{extramarks}
\usepackage{amsmath}
\usepackage{amsthm}
\usepackage{amsfonts}
\usepackage{siunitx}
\usepackage{tikz}
\usepackage[plain]{algorithm}
\usepackage{algpseudocode}
\usepackage{multirow}
\usepackage{booktabs}
\usepackage{graphicx}
\usepackage{subfigure}
\usepackage{mathrsfs}
\usepackage[colorlinks,linkcolor=black,anchorcolor=black,citecolor=black,urlcolor=blue]{hyperref}
\usepackage{amsmath,bm}
\usepackage{booktabs}
\usepackage{mathtools}
\usepackage{amssymb}
\usepackage{caption}
\usepackage{hyperref}
\usepackage{capt-of}
\usepackage{mciteplus}
\usepackage{cite}
\usepackage{mathrsfs}
\usepackage{enumitem}
\usepackage[title,titletoc,toc]{appendix}
\usepackage{xr}
\usepackage{parskip}
\usepackage{soul}
\usepackage{textcomp}
\usepackage[colaction]{multicol}
\usepackage[switch]{lineno}
\usepackage{lipsum}
\usepackage{etoolbox}
\usepackage{longtable}
\usepackage{array}
\usepackage{tablefootnote}
\usepackage{ragged2e}
\usepackage{soul}
\newcolumntype{C}[1]{>{\centering\arraybackslash}p{#1}}
\usetikzlibrary{automata,positioning}
\usepackage{xr}
\usepackage{xr-hyper}
\usepackage{import}
\usepackage{cleveref}
\allowdisplaybreaks
\makeatother

\usepackage{url}
\usepackage{todonotes}

\usetikzlibrary{automata,positioning}

\newtheorem{thm}{Theorem}[section]

\newtheorem{Def}{Definition}[section]

\newtheorem{cor}{Corollary}[section]

\begin{document}	

	\title{Subspace Tensor Orthogonal Rotation Model (STORM) for Batch Alignment, Cell Type Deconvolution, and Gene Imputation in Spatial Transcriptomic Data}
	
	\author{Sean Cottrell$^{1,2}$, Guo-Wei Wei$^{1,3,4}$\footnote{
			Corresponding author.		Email: weig@msu.edu}, Longxiu Huang$^{1,2}$\footnote{
			Corresponding author.		Email: huangl3@msu.edu} \\
		\\
		$^1$ Department of Mathematics, \\
		Michigan State University, East Lansing, MI 48824, USA.\\
		 $^2$ Department of Computational Mathematics, Science, and Engineering, \\
		Michigan State University, East Lansing, MI 48824, USA.\\
		$^3$ Department of Electrical and Computer Engineering,\\
		Michigan State University, East Lansing, MI 48824, USA. \\
		$^4$ Department of Biochemistry and Molecular Biology,\\
		Michigan State University, East Lansing, MI 48824, USA. 
	}
	\date{\today} % Date for the report
	
	\maketitle
	
	\begin{abstract}

Spatial transcriptomics data analysis integrates cellular transcriptional activity with spatial coordinates to identify spatial domains, infer cell-type dynamics, and characterize gene expression patterns within tissues. Despite recent advances, significant challenges remain, including the treatment of batch effects, the handling of mixed cell-type signals, and the imputation of poorly measured or missing gene expression. This work addresses these challenges by introducing a novel Subspace Tensor Orthogonal Rotation Model (STORM) that aligns multiple slices which vary in their spatial dimensions and geometry by considering them at the level of physical patterns or microenvironments. To this end, STORM presents an irregular tensor factorization technique for decomposing a collection of gene expression matrices and integrating them into a shared latent space for downstream analysis. 
%The method incorporates gene–gene interactions across samples and leverages spatial adjacency among cells both within and between samples to smooth non-biological variation. Furthermore, it exploits a hypothesized redundancy in gene expression at the microenvironment level to model a low-dimensional latent gene field within slice-specific spatial subspaces which can be aligned across samples via a rotation. 
In contrast to black-box deep learning approaches, the proposed model is inherently interpretable. Numerical experiments demonstrate state-of-the-art performance in vertical and horizontal batch integration, cell-type deconvolution, and unmeasured gene imputation for spatial transcriptomics data. 
	
	\end{abstract}
Keywords: {Spatial Transcriptomics, Tensor Decomposition, Batch Integration, Gene Imputation, Cell Type Deconvolution}
	
	%\newpage
	%
	%{\setcounter{tocdepth}{4} \tableofcontents}
	%\setcounter{page}{1}
 %
	\newpage

	\section{Introduction}

The analysis of gene expression data has seen a series of significant innovations in the last decade. As such, it currently plays a pivotal role in biological and medical research \cite{Jain2024SpatialTranscriptomicsReview}. In recent years, spatial transcriptomics, in particular, has emerged as a powerful tool in advancing our understanding of the relationship between the gene expression of cells and their spatial distributions, which underpins tissue pathology and function annotation. The applications of this spatial information range from inferring cell-cell communication, trajectory inference, spatially variable gene detection, to spatial domain segmentation \cite{st_methods}. Current spatial transcriptomics technologies are varied. In-situ hybridization techniques use fluorescent dye labeled probes hybridized to specific RNA transcripts to measure gene expression activity at subcellular resolution, which has evolved from low-gene-throughput single-molecular FISH (smFISH) to high-gene-throughput multiplexed error robust FISH (MERFISH) and sequential FISH (seqFISH and seqFISH+) technologies. More recently, researchers have developed in situ capturing methods to perform RNA sequencing of the whole transcriptome with positional barcodes in a spatial genomic array aligned to locations on the tissue. These methods range from lower resolution Spatial Transcriptomics (ST) to higher resolution Slide-seq, or even sub-cellular resolution technologies such as Stereo-seq. However, the continued development of ST techniques has created challenges for ST data analysis.  One of these challenges is due to diverse sources and technologies, leading to the so called  batch effects, which prevent the reliable extraction of genuine biological signal present in the  samples. Another challenge is mixed cell-type signals in ST data because low-resolution individual spots frequently capture transcripts from multiple neighboring cell types. This overlap leads to composite gene expression profiles that obscure true cell-type–specific signals. Consequently, effective deconvolution of mixed cell-type signals is needed for accurate biological interpretation. The other challenge is poorly measured gene expression in ST data due to technical instability, noise, limited sequencing depth, and platform-specific detection biases,  resulting in missing or unreliable expression values that distort spatial patterns and downstream analyses, calling for robust modeling and gene imputation strategies  to recover biologically meaningful signals and improve analytical accuracy.

Much attention has been paid to the integration and modeling of single cell and spatial transcriptomics data \cite{Zhang2025.03.12.642755,cottrell2025multiscale}. A common strategy in data alignment is through an analysis in a shared embedding space, aiming to balance the removal of batch effects with the conservation of nuanced biological information. Harmony \cite{Korsunsky2019Harmony} groups cells into multi-dataset clusters in a joint PCA embedding space, using soft clustering to assign cells to potentially multiple clusters to account for smooth transitions between cell states. STAligner \cite{Zhou2023STAligner} uses the gene expression information and spatial coordinates from multiple ST slices as its input to construct a spatial neighbor graph between spots. Based on these graphs, STAligner uses a graph attention auto-encoder neural network to learn spatially aware embeddings. Spiral \cite{Guo2023SPIRAL} is composed of four neural networks for batch integration to sequentially encode gene expressions and spatial coordinates into a low-dimensional latent space and to reconstruct gene expressions.  Tensor structures have been explored in ST data analysis previously for a variety of tasks. For example, Song et al. introduced a graph-guided neural tensor decomposition utilizing a three layer neural network to learn non-linear relations among all the elements in each mode (spatial and gene) for constructing the factors of a canonical polyadic decomposition (CPD). The reconstructed tensor can then be used to infer denoised gene expression values \cite{Song2023GNTD}. In terms of irregular tensor structures, Ramirez et al. utilized a PARAFAC2 model to jointly embed perturbed and unperturbed gene expression matrices for predicting perturbation responses \cite{Ramirez2025PARAFAC2RISE_bioRxiv}. These approaches demonstrate the utility of shared embeddings and multi-way structure, but a unified tensor framework for multi-sample ST that addresses batch effects remains lacking.

In the context of tensor based modeling, the objects of interest become collections of gene expression matrices, typically irregular in size and often differing in spatial geometry, rendering standard tensor-based models ineffective or inapplicable. However, spatial autocorrelation in gene expression creates redundancy, as neighboring locations that share coherent transcriptional patterns often organize into microenvironments. Importantly, a microenvironment is not merely a property of a single 2D section but a 3D tissue entity. Each ST slice may capture a 2D cross-section of the same underlying pattern. Multi-slice integration can therefore be framed as discovering and comparing shared microenvironment patterns repeatedly sampled across slices. The consideration of microenvironment eliminates the irregularity in the cell mode of the gene expression tensor, making it amenable to classical tensor factorization for joint embedding space representations. However, this approach alone does not automatically resolve batch effects. To this end, one needs explicit cross-sample constraints in the shared embedding space to align microenvironments and suppress batch-driven variation while preserving the genuine biological structure. 

In this study, we address these challenges using a Subspace Tensor Orthogonal Rotation Model for Spatial Transcriptomics (STORM). Mathematically, a collection of single cell or spatial transcriptomics experiments naturally exhibits an irregular 3-way tensor structure, with tensor modes along the spatial and gene dimensions as well as a third mode specifying the slice/sample/batch.  For each slice, batch, or sample, we identify a common gene set and eliminate irregularity in the spatial mode by integrating slices at the level of microenvironments, or spatial subspaces. The orthonormal frames of these subspaces rotate samples into aligned coordinates, which mitigates batch effects across samples. The resulting `subspace tensor' of microenvironments is then factorized.  First, a shared Gram matrix (up to scaling) across slices is enforced. This constraint ensures a common covariance structure among microenvironments and further reduces batch effects. Gene expression profiles are then linearly mapped into these subspaces. The orthonormal frames subsequently reconstruct the low-dimensional spatial structure at the original spatial resolution for downstream analysis. 
We further impose Cartesian product graph smoothing during tensor reconstruction to promote similarity among functionally related genes, spatially adjacent spots, and coupled spots across samples. This mathematical framework provides a powerful solution for both vertical and horizontal integration of spatial transcriptomic data. When paired with a single-cell RNA-seq reference, it also enables inference of cell-type composition in spatial spots and imputation of missing gene counts. Extensive benchmarking demonstrates that STORM achieves state-of-the-art performance in spatial transcriptomic data analysis.
    
\section{Results}\label{sec:results}
  \subsection{Overview of Subspace Tensor Orthogonal Rotation Model (STORM)} 
 STORM learns an interpretable shared latent representation for collections of spatial transcriptomics (ST) slices by integrating (i) slice-adapted spatial subspaces that capture structures at the microenvironment-level, (ii) spatial proximity within and across samples, and (iii) functional gene relations encoded by protein-protein interactions. Figure~\ref{fig:STORM} provides a high-level overview of the workflow and the outputs used in the following evaluations; methodological details and the solver
are described in Section~\ref{sec:method}.  We evaluate STORM on four core ST analysis tasks: horizontal and vertical batch integration, cell-type deconvolution using a single-cell reference, imputation of unmeasured genes in imaging-based ST, and characterization of temporal progression across developmental stages. We show that STORM yields robust performance and strong biological fidelity across datasets, platforms, and tissues. The details of the data considered in this study are summarized in the Supplementary Information Section 5. 
    
    \begin{figure}[H]
        \centering
        \includegraphics[width=\linewidth]{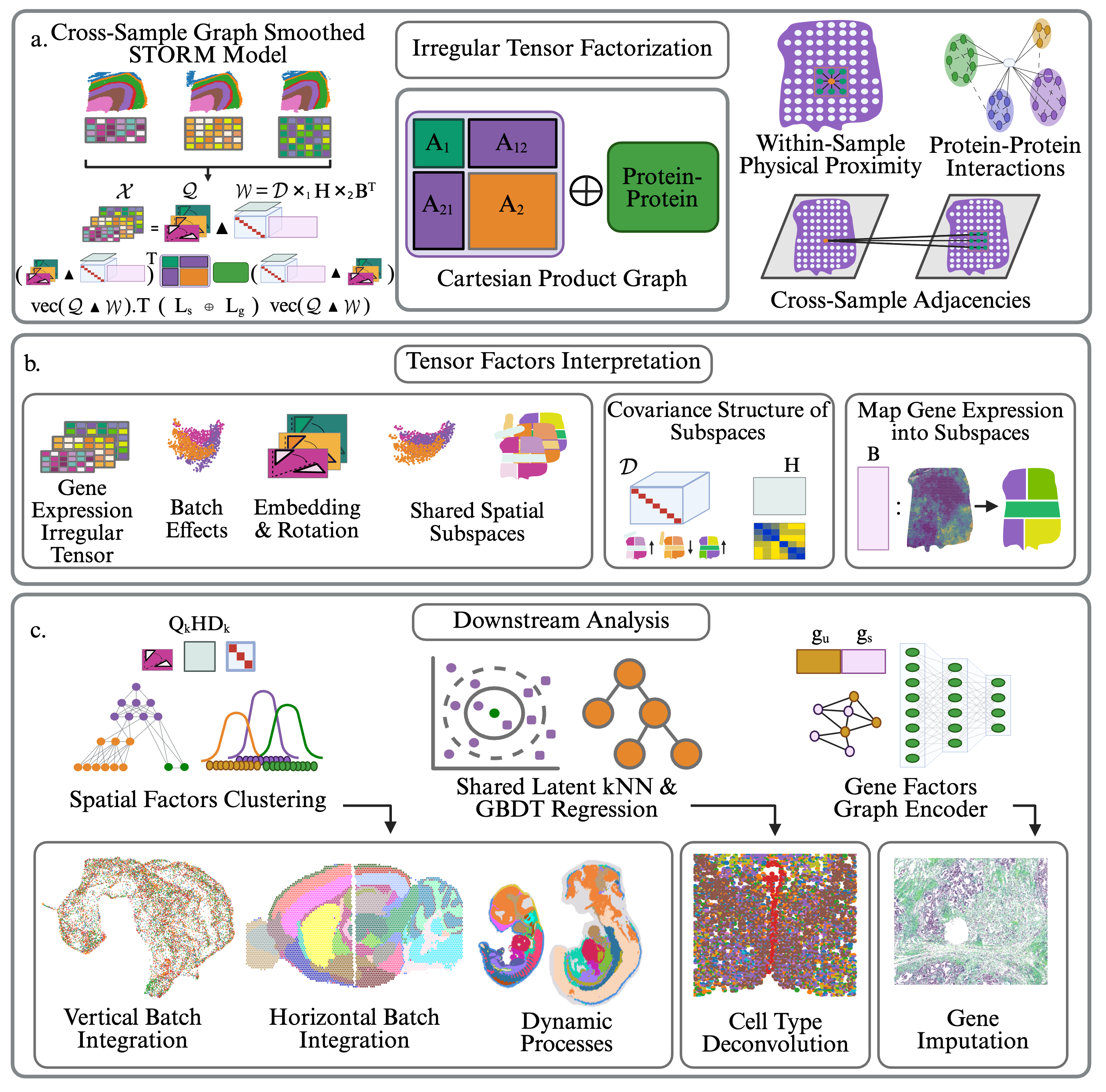}
        \caption{Overview of the STORM framework.
a. Given multi-slice (multi-batch) spatial transcriptomics data, we represent each slice $k=1,...,n_b$ as a spot-by-gene expression matrix $X_k \in \mathbb{R}^{n_{s_k} \times n_g}$ of $n_{s_k}$ spatial spots and $n_g$ genes.  STORM learns, for each slice, an orthonormal spatial alignment matrix $Q_k\in \mathbb{R}^{n_{s_k} \times R}$  that acts only on the spot (spatial) mode and rotates slice-specific spatial subspaces into an $R$ dimensional shared spatial space. In this space, we obtain a regular subspace tensor $\mathcal{W} \in \mathbb{R}^{R \times n_g \times n_b}$ of ST slices which we can factorize. 
b. In this aligned space, STORM estimates shared factors $H\in \mathbb{R}^{R \times R}$ and $B\in \mathbb{R}^{n_g \times R}$ that capture common latent structure across slices, together with a tensor of diagonal scaling matrices $\mathcal{D}\in \mathbb{R}^{R \times R \times n_b}$, such that $\mathcal{D}_{::k} = \text{diag}(w_k)\in \mathbb{R}^{R \times R}$ ($w_k \in \mathbb{R}$) that modulates the contribution of the shared components within each slice. The reconstructed expression  profiles are further refined using spatial proximity, protein-protein interaction (PPI) gene relations, and cross-sample spot  adjacencies. 
c. The resulting shared embedding $Z_k = Q_kHD_k$ supports downstream analyses including data integration, cell-type deconvolution, and gene imputation.} 
        \label{fig:STORM}
    \end{figure}
    
    % Figure \ref{fig:STORM} provides an overview of the STORM framework. The integrative subspace tensors approach begins by finding a series of frames for low dimensional subspaces of the original spatial space. Gene expression activity is mapped into these subspaces, and the pairwise gene factor relations and scalings are iteratively updated to approximate the original gene expression information. The reconstruction is then smoothed across the Cartesian product of a spatial / cross-batch graph and a PPI graph, achieving an effective shared latent space representation of the gene expression data. This latent space structure can then be used for, say, spatial domain detection via clustering. Alternatively, factorizations of ST samples and sc-RNA-seq samples in this joint space can be used to infer cell type compositions of spots as well as imputing unmeasured gene counts.  
    
    \subsection{STORM Effectiveness in Horizontal and Vertical Integration of Batch Affected Tissue Samples}

    We first comprehensively evaluate the effectiveness of STORM on a cross-sample alignment for vertical and horizontal batch integration. In the case of vertical batch integration, we are faced with the issue of batch effects due to technological and experimental setting differences between runs. These batch effects can blur genuine biological variation in the data, and so limit the accuracy of any multi-sample analysis. In Figure \ref{fig:integration}c, we see this challenge presenting itself in the popular Dorsolateral Prefrontal Cortex (DLPFC) Data from 10X Visium. A naive joint PCA achieves a very poor integration of four samples, and spatial domains are blurred and poorly separated from one another. The LIBD dorsolateral prefrontal cortex data contains the spatially resolved transcriptomic profiles of 12 slices as well as their manual annotations of neuronal layers and white matter \cite{dlpfc}. Specifically, the 151673, 151674, 151675, and 151676 samples are each taken from the same donor and so can be used to evaluate the efficacy of various integration methods. 

    We  evaluate integration effectiveness based on a combination of two important metrics. First, we  evaluate the integrated spatial domains predicted by each method via the Adjusted Rand Index (ARI) of the clusters. Meanwhile, a composite metric called the F1-harmonized Local Inverse Simpson Index (F1LISI) can be used to quantify the balance between batch effect removal and biological conservation. The F1LISI index (ranging from 0 to 1) integrates batch-grouped LISI (LISIbatch) and domain-grouped LISI (LISIdomain) through a tunable weighting coefficient $\alpha$, with values closer to 1 indicating superior technical noise elimination while also preserving biological variance \cite{Tran2020BatchCorrectionBenchmark}.

    The performance of each method with respect to the aforementioned  metrics can be seen in Figure \ref{fig:integration}d, where one can note that STORM is able to achieve the highest measure of spatial domain detection accuracy according to ARI. Across the integration of these four samples, STORM achieves a mean ARI of 0.626, outperforming the next best method, STAligner \cite{Zhou2023STAligner}, by 7.6\% in this respect. In terms of integration quality, STORM displays state-of-the-art results with a median F1LISI index score of 0.932, 6.75\% higher than GraphST \cite{Long2023GraphST}, the next best method for this metric. We emphasize that STORM is capable of simultaneously achieving an accurate spatial domain detection alongside an effective batch integration. Methods such as STAligner and Splane \cite{Xu2023SPACEL} which perform well in terms of ARI achieve very low F1LISI median scores while GraphST achieved a higher median F1LISI score, but the resultant spatial domain detection was inaccurate. Spiral is the only other method capable of performing well in both cases, though still trails STORM by 17.3\% and 11.2\% in terms of ARI and median F1LISI respectively. Other methods, such as DeepST \cite{Xu2022DeepST} and SEDR \cite{Xu2024SEDR} also exhibit considerable trade-offs in their integration quality versus clustering performance. A full description of each method's performance can be found in the Supplementary Information Section 6. 
    
    \begin{figure}[H]
        \centering
        \includegraphics[width=\linewidth]{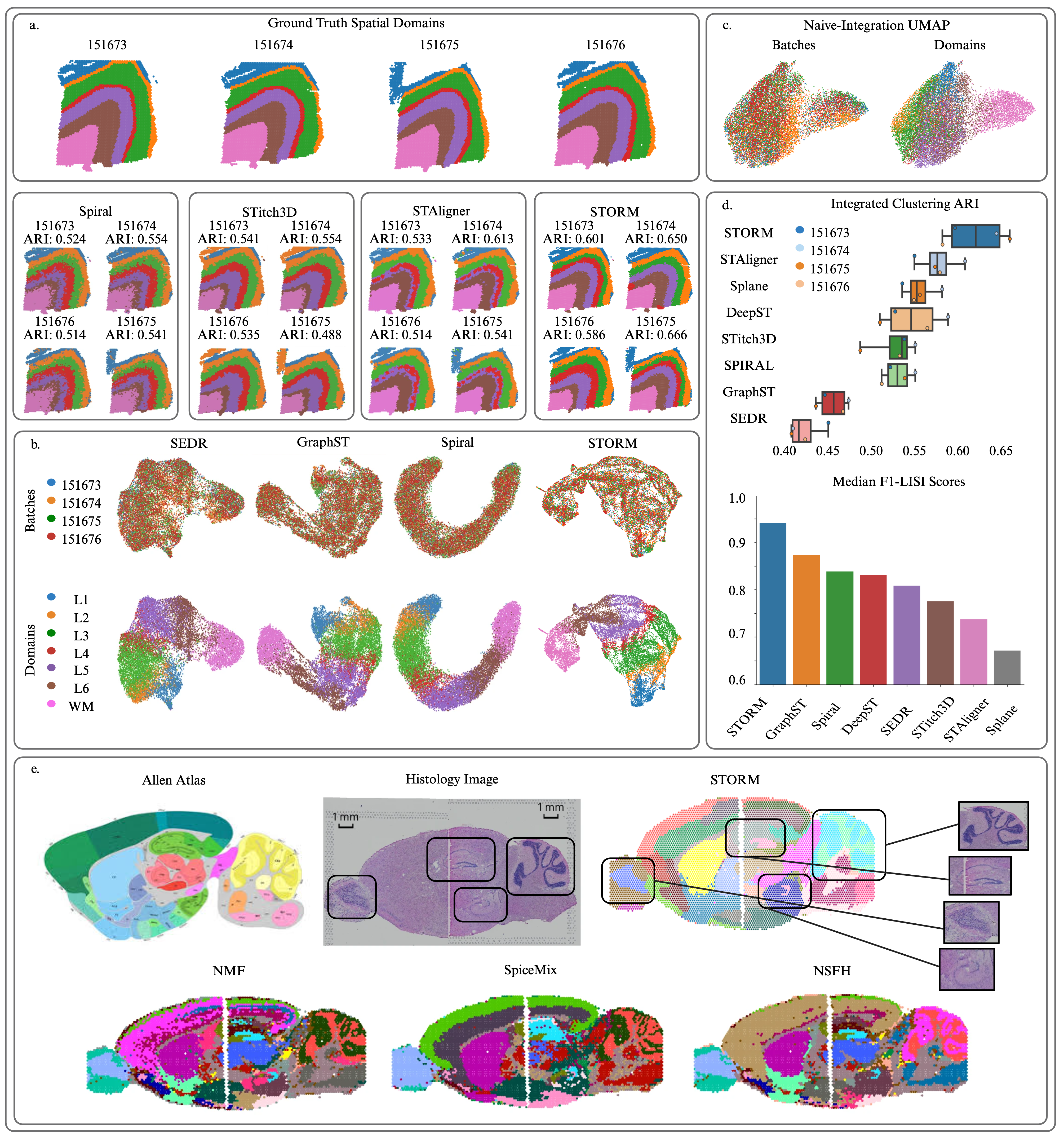}
        \caption{Illustration of STORM clustering. a. Ground truth annotations of DLPFC samples 151673, 151674, 151675, and 151676 into innermost white matter region and 6 linearly outward extending neuronal layers. We additionally present comparisons of various integrated spatial domain detections on these samples against the STORM induced domains as well as the ground truth. We note the strong correspondence between the STORM domains and the ground truth relative to other methods. b. Integrated UMAP embeddings of each method colored by sample and ground truth spatial domain. Coloring by batches provides a qualitative measure of the integration of each sample by each method. Coloring by ground truth spatial domain provides a measure of the preservation of genuine biological variation across the samples for each method. c. UMAP embeddings induced by joint PCA on the concatenated gene expression matrices shows poor integration of samples as well as poor separation of spatial domains- illustrating the presence of batch effects and motivating additional computational techniques. d. Quantitative comparison of the performance of various methods for clustering and integration of the DLPFC data according to ARI and F1LISI metrics. STORM achieves the highest performance according to both metrics. e. Horizontal integration of the anterior and posterior sections of a mouse brain. We display the histology image alongside the ground truth Allen Mouse Brain annotations and compare spatial domain detections of various methods against STORM. }
        \label{fig:integration}
    \end{figure}

    In Figure \ref{fig:integration}a, we visually examine the spatial domains predicted by each method. The ground truth tissue architecture is comprised of an innermost White Matter region with 6 neuronal layers extending outward linearly. STitch3D \cite{Wang2023STitch3D} exhibits an inability to reliably separate the outer neuronal Layers 1 and 2. STAligner, meanwhile, shows domain discontinuities in Layers 4 and 5. Spiral achieves a fairly accurate linear layer structure, though its domains also exhibit discrete anomalous spots, indicating poor separation in the latent space. This is confirmed in the shared UMAP plot of the Spiral embedding in Figure \ref{fig:integration}b, where ground truth cluster boundaries do not exhibit clean separation. This is reflected in Spiral's relatively lower median F1LISI score. In this same plot, we see that GraphST, meanwhile, shows an impressive integration of batches, but also cannot reliably separate the true clusters. This results in very poor spatial domain detection performance. Our STORM cleanly partitions the regions of the tissue, though in the case of Sample 151674 it cannot properly identify the thin boundary between inner and outer neuronal layers, and instead attempts to segment the innermost Layer 6 at the tissue edge into a separate domain. In spite of this, it still achieves the highest ARI on this sample by very faithfully reconstructing the other domains. Its UMAP plot also exhibits a meaningful integration of batches while retaining reliable separation between the clusters. 

    In Figure \ref{fig:integration}e, we further display the performance of STORM on a cross-sample horizontal integration of an Anterior and Posterior mouse brain sample from 10x Visium. We compared STORM with a baseline joint NMF, as well as the NSFH and SpiceMix methods \cite{Chidester2023SPICEMIX,Townes2023NSF}. As tissue samples can be significantly larger than the capture slides used for spatial transcriptomics, horizontal integration enables the data from multiple capture slides to be stitched together while ameliorating the issue of batch effects between experimental runs. We compare the spatial domain detections on this integrated sample with the manually annotated Allen Mouse brain atlas \cite{allen}. For these methods, most of the clusters show general agreement with the true tissue structure. However, NSFH failed to properly identify the simple lobule and  granular layers. SpiceMix, meanwhile, could not reliably identify the field CA1 pyramidal layer. The NMF clusters display considerable fragmentation, reflecting the lack of local spatial information in its modeling framework. All these three methods struggled in identifying the globus pallidus region. Our STORM, meanwhile, was able to accurately segment each of these structures. Interestingly, STORM is also the only method that distinguishes the anterior somatosensory areas from the posterior parietal association and visual areas. These results demonstrate STORM's ability to integrate large tissue samples in such a way that cross-sample regions are reliably identified while genuine biological variation between the samples does not become blurred. 
    
    \subsection{STORM Accurately Deconvolves Cell Type Compositions in Low Resolution Spatial Transcriptomics Experiments}

    Spatial transcriptomics platforms are often limited in their spatial resolution, and so the resultant data often contains spots with sizes in the range of 55 micrometers or more in diameter, e.g., the 10X Visium platform \cite{dlpfc}. This  means that for such technologies, each sample is in fact comprised of multiple cells or partial fractions of cells. A practical goal is then to predict the cell type compositions of these spots using a single-cell resolution reference dataset to demystify genuine cell type specific signals. Our STORM accomplishes this goal in two ways. The tensor factorization embeds a spot resolution spatial dataset into a shared latent space with a single cell reference. We can then use a GBDT to learn a mapping from the embedding coordinates to cell type distributions, or we can utilize k-Nearest-Neighbors (k-NN) to infer the cell types in spots from their latent space single cell neighbors. We demonstrate the effectiveness of both approaches in this section. To assess the performance of STORM in this task, we first compared our method against the results from a popular cell type deconvolution benchmarking paper \cite{Li2023DeconvolutionBenchmark}. This study compared the deconvolution performance of 18 methods on 13 synthetic datasets from the popular seqFISH+ and MERFISH platforms as seen in Figure \ref{fig:deconvolution}a. Synthetic data was constructed via taking a ground truth single-cell resolution spatial transcriptomics dataset and binning cells together to simulate low spatial resolution, i.e. spots with diameters in the range of 55 micrometers. The original single-cell resolution gene expression values can then be used as a non-spatial single cell reference.

    Figure \ref{fig:deconvolution}b illustrates the results of this comparison. The benchmarking study looked at various popular methods such as SpatialDWLS \cite{Dong2021SpatialDWLS}, CARD\cite{Ma2022CARD}, RCTD\cite{Cable2022RCTD}, Stereoscope\cite{Andersson2020Stereoscope}, STRIDE\cite{Sun2022STRIDE}, Tangram\cite{Biancalani2021Tangram}, Cell2Location\cite{Kleshchevnikov2022Cell2location}, DSTG\cite{Song2021DSTG}, SD2\cite{Li2022SD2}, SpaOTsc\cite{Cang2020SpaOTsc}, SPOTlight\cite{ElosuaBayes2021SPOTlight}, Berglund, NMFreg, SpatialDecon\cite{Danaher2022SpatialDecon}, DestVI\cite{Lopez2022DestVI}, novoSpaRc\cite{Moriel2021NovoSpaRc}, SpiceMix\cite{Chidester2023SPICEMIX}, and STdeconvolve\cite{Miller2022STdeconvolve}. It measured the Root Mean Squared Error and Jensen Shannon Divergence between the predicted cell type compositions of each spot compared to the ground truth original data. Scores were then averaged over each cell type to obtain an overall performance score for each method on each metric. Following this prescribed approach, we found that STORM was able to outperform all other methods considered in this study. Specifically, on the seqFISH+ data STORM achieved an average RMSE of 0.132, 33\% better than the 0.198 score of the next best SpatialDWLS. In terms of Jensen Shannon Divergence, STORM outperformed SpatialDWLS by 14.4\%, achieving an average score of 0.089. On the synthetic MERFISH data, STORM achieved a score of 0.147 and 0.156 in these metrics, outperforming the next best Cell2Location and DestVI respectively. In Supplementary Information Section 1, we manually examine the spatial distributions of various cell types against the ground truth distributions in this data. We observe that the predicted spatial mapping of cell types is generally very accurate in both data, highlighting the impressive performance of STORM in this task.
    
    \begin{figure}[H]
        \centering
        \includegraphics[width=\linewidth]{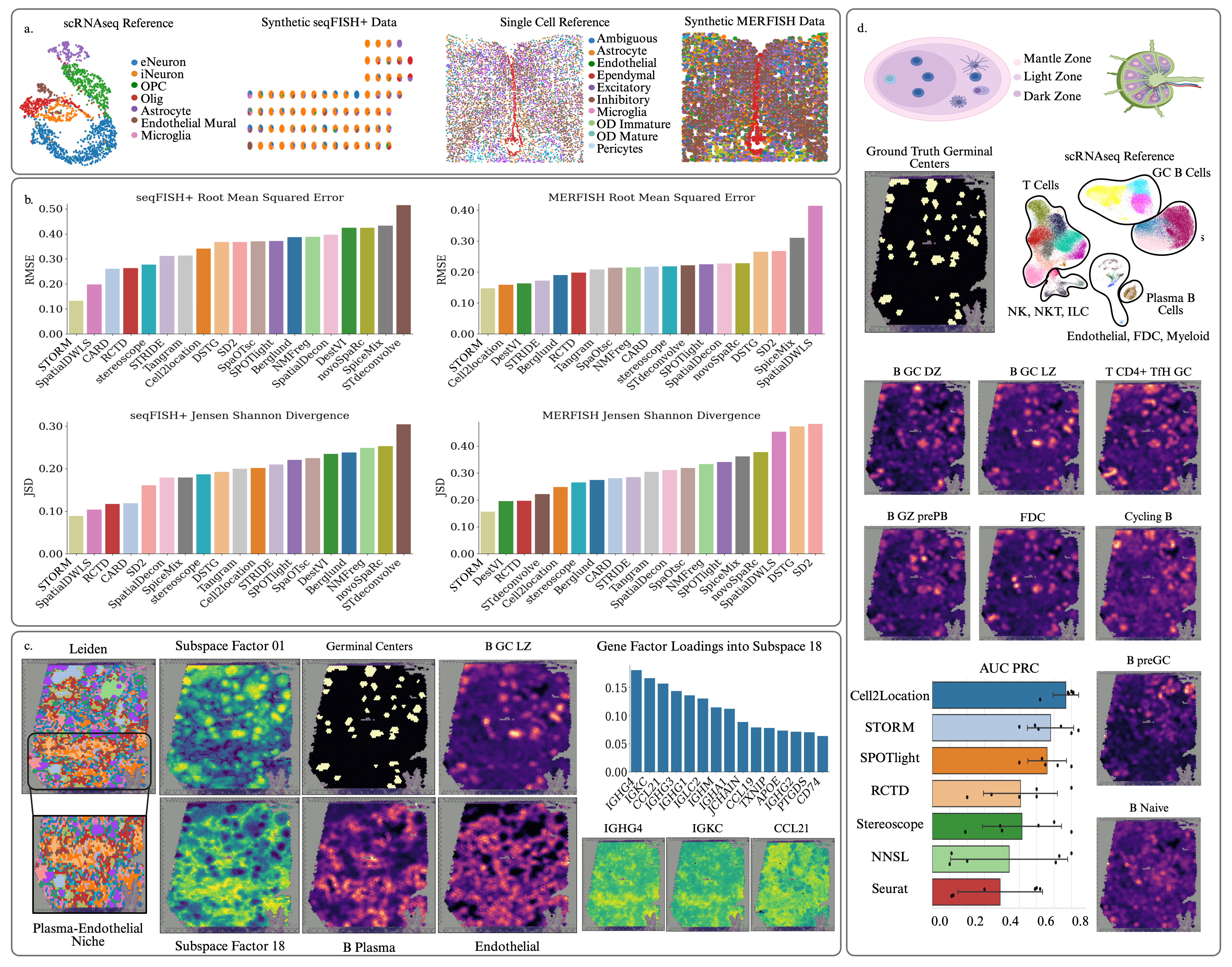}
        \caption{Illustration of STORM for cell type deconvolution. a. Synthetic low resolution seqFISH+ and MERFISH data is created by binning the cells in single-cell resolution ST data into multi-cell spots for convenient benchmarking. Both data types were binned to create synthetic spot sizes of 50 micrometers- matching the resolution of technologies such as Visium. b. Quantitative comparisons of cell type deconvolution performance of STORM against recent other methods taken from benchmarking results on the seqFISH+ and MERFISH synthetic data according to the JSD and RMSE metrics. c. Visualization of the spatial distribution of STORM induced spatial subspace factors $Q_kHD_k$ with Leiden spatial domains and cell type spatial mappings. The subspaces, clusters, and cell type proportions converge on coherent biological niches, such as germinal centers and plasma-endothelial microenvironments. Gene factor loadings in $B$ into these subspaces support a coherent biological narrative. This demonstrates the interpretatability of spatial factors as spatial niches. d. Spatial mapping of various B cells, FDCs, and T cells to germinal center zones of human lymph node sample by STORM as well as a quantitative comparison of STORM's performance against other methods in the task of mapping germinal center specific cells. }
        \label{fig:deconvolution}
    \end{figure}

    In Figure \ref{fig:deconvolution}c and d, we turn to several case studies examining the spatial mapping of single cells to low resolution Visium Human Lymph Node samples. In Supplementary Information Section 1 we perform a similar analysis on spatially mapping excitatory neurons and oligos cells to the dorsolateral prefrontal cortex. First, following the procedure outlined in the Cell2Location study \cite{Kleshchevnikov2022Cell2location}, we applied our STORM to spatially map tissue microenvironment of the human lymph node. The human lymph node is characterized by dynamic microenvironment. We jointly analyzed a Visium dataset of the human lymph node from 10x Genomics and spatially mapped a comprehensive atlas composed of 73,260 cells and 34 cell types. This data was derived by an integration of multiple sc-RNA-seq datasets from human secondary lymphoid organs. The Cell2Location study observed that a histological examination of the lymph node Visium sample revealed multiple germinal centers. Figure \ref{fig:deconvolution}d provides a graphical model of the general structure of germinal centers, which are comprised of a mantle zone, or a ring of mostly resting, naive B cells surrounding the GC, a light zone containing loosely packed post-maturation B cells along with follicular dendritic cells and Tfh cells, and a dark zone packed with rapidly dividing B cells. Utilizing a GBDT to map the distributions of cell types onto the spatial locations via STORM, we observe that the groups of cell types corresponded accurately to known functionally relevant cellular compartments of the lymph node. We reliably identified the aforementioned GC dark zone, where B cells proliferate along with the GC light zone, where they undergo selection by T follicular helper cells and FDCs to differentiate into antibody-producing plasmablasts. Quantitatively, we found that in terms of AUC PRC score, our cell type mapping was superior to other popular methods from the literature and generally competitive with Cell2Location. Additionally, STORM demonstrated an ability to spatially map the preGC B cell and Naive B cell populations in such a way that they occupy different tissue compartments from the upstream activated cells. Specifically, they seem to generally surround the GC regions while not being highly expressed within these regions, reflecting known biology.

    We next further analyze the STORM induced spatial subspace factors in this sample, and observing how they align with genuine biological niches.  In the top row, we highlight a spatial subspace factor which clearly aligns itself to the lymph node germinal centers, and we have noted that the single cells which spatially mapped to these regions accurately reflect what one would expect of this microenvironment structure. In the bottom row, we highlight a spatial subspace, which appears as a contiguous region that has high mapped abundance of B plasma cells and endothelial cells. Biologically, a plasma–endothelial niche in a lymph node is plausible because plasma cells/plasmablasts are not randomly distributed. Endothelial and associated stromal networks are major organizers of where immune cells sit, via chemokines/adhesion cues. Therefore, plasma enrichment coinciding with a vascular/stromal signature is highly plausible \cite{Skokos2010}. We additionally examine which genes had the highest factor loadings into this subspace. IGHG4, IGHG3, IGHG1, IGHG2, IGHA1, IGHM are heavy chains while IGKC, IGLC2 are light chains which make up the core of antibodies. JCHAIN is a hallmark of antibody-secreting cells \cite{CastroFlajnik2014}. CCL21 and CCL19 are chemokines produced largely in lymph node stromal compartments that organize immune cell trafficking and spatial zoning. Their presence suggests the factor is not only plasma, but plasma in a region defined by an endothelial guidance program. CXCL9 is typically associated with inflammatory chemokines that can show up in tumor microenvironment \cite{Brown2015}. Examining the spatial gene maps of IGHG4, IGKC, and CCL21, we note IGHG4/IGKC peaking where plasma B cell signal is high, while CCL21 generally outlines the endothelial guidance network. Taken together they mirror the subspace factor’s footprint, indicating its correspondence with a genuine spatial niche. The Leiden clustering of these spatial subspace factors also is spatially coherent and concordant with both the inferred cell-type maps from the single cell reference and factor-specific gene programs. In Supplementary Information Section 1 we expand on this analysis by analyzing the co-loadings of genes into each spatial subspace, as well as spatially plotting all 30 factors and observing their spatial localization which again corresponds with various cell niche structures. 
    
    \subsection{STORM Imputes Non-Measured Genes in Imaging-Based ST Data from sc-RNA-seq Reference}
    While some spatial transcriptomics platforms, such as several imaging based technologies, can achieve single cell or even sub-cellular spatial resolution, there is generally a trade off between spatial resolution and gene targeting depth, with many techniques not able to sequence the whole genome. Additionally, sequencing mechanical errors may create dropout events, where a measured gene shows zero activation in spite of the gene possibly having actually been transcribed. For this reason, gene imputation emerges as an important problem in ST data analysis. To this end, we combine our STORM joint dimensionality reduction framework with a powerful graph neural network architecture to impute the measurements of unmeasured genes. Figure \ref{fig:imputation}a demonstrates the STORM workflow for gene imputation. Specifically, we utilize STORM to jointly embed a spatial dataset with low gene sequencing depth and a full genome coverage single cell RNA-seq dataset. Additionally, the STORM gene factor loadings provide a data and protein interaction informed set of gene relations. We can then construct a 3-way tensor of genes, spots, and sc-RNA-seq samples which can be smoothed along its axes by a GNN. A final linear layer then predicts the final gene expression vectors for each spot. More details can be found in the Methods section. 

    For validating our technique we tested STORM against several popular methods on a series of several datasets. The datasets were normalized according to the procedure outlined in the SpaGE study \cite{Abdelaal2020SpaGE}. Specifically, a dataset underwent normalization, involving the division of counts within each cell by the total number of transcripts. This result was then scaled by the median number of transcripts per cell and log-transformed with a pseudo-count. The examined dataset pairs exhibit a wide array of levels of gene detection sensitivity, sample sizes, and the number of spatially measured genes. The sizes and statistics of each dataset can be found in the Supplementary Information Table 1. 
    
    \begin{figure}[H]
        \centering
        \includegraphics[width=\linewidth]{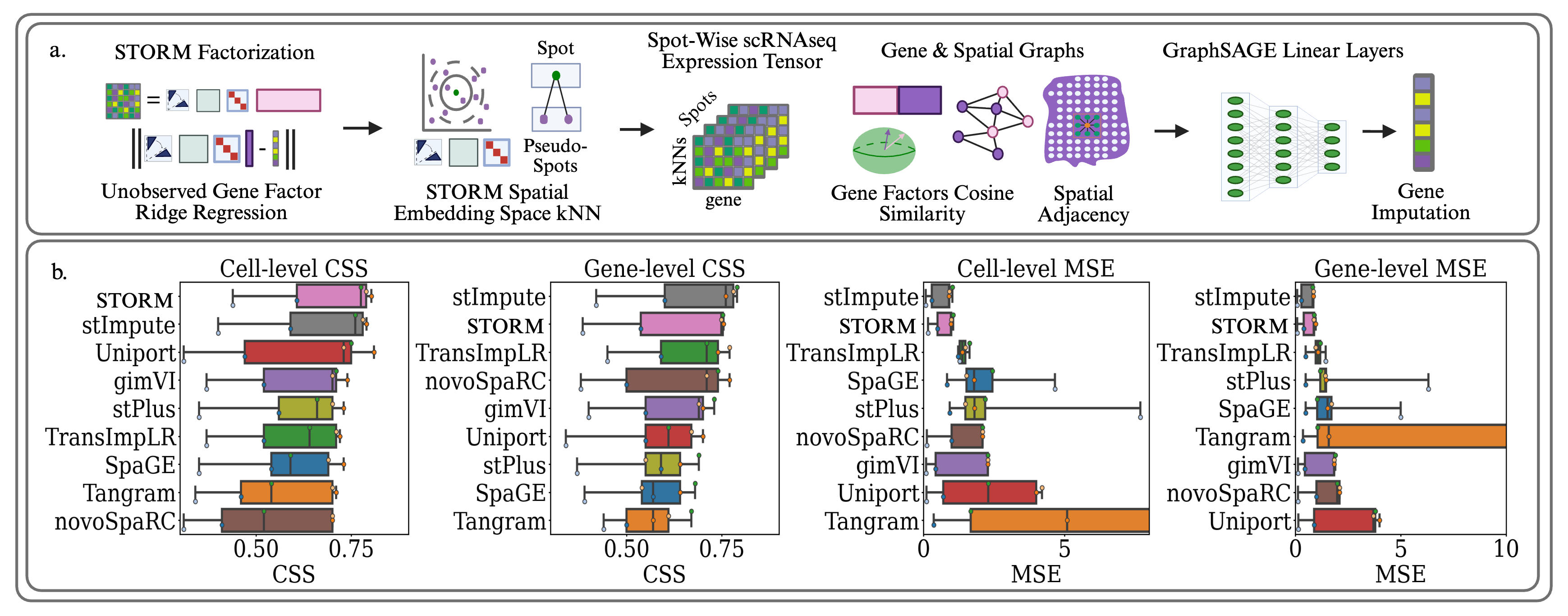}
        \caption{a. Overview of STORM workflow for gene imputation. STORM is used to construct a shared latent space between a ST sample and sc-RNA-seq reference in $Q_kHD_k$. This shared latent space is used to construct a 3-way sc-RNA-seq reference tensor of spatial spots and their neighboring single cell pseudo-spots, which can be input to a GraphSAGE model to predict imputed gene expression values for each ST cell from the single cell reference. A gene graph is constructed from the completed gene factor loadings in $B$ to further smooth the GraphSAGE model along this axis. b. Quantitative comparison of STORM against several other imputation methods on 5 benchmark datasets. Methods were measured on RMSE and CSS at the cell and gene levels for thoroughness of evaluation. STORM consistently ranks among or as the top performer. }
        \label{fig:imputation}
    \end{figure}
    
    Figure \ref{fig:imputation}b illustrates the gene imputation performance of STORM against 8 other state-of-the-art gene imputation techniques (stImpute \cite{Zeng2024stImpute}, Uniport \cite{Cao2022uniPort}, gimVI, stPlus \cite{Chen2021stPlus}, TransImpLR \cite{Qiao2024TransImpute}, SpaGE \cite{Abdelaal2020SpaGE}, Tangram \cite{Biancalani2021Tangram}, and novoSpaRC \cite{Moriel2021NovoSpaRc}) on 5 spatial transcriptomics and sc-RNA-seq pairs. We specifically examined imaging-based spatial transcriptomic (ST) technologies such as SeqFISH+, osmFISH, STARmap, and MERFISH, which generally achieve a limited number of genes from the entire transcriptome. In keeping with previous studies such as stImpute, Tangram, and TransImpLR, we looked at four metrics to evaluate the imputation performance. Specifically, we focus on the Cosine Similarity Score (CSS) and Mean Squared Error (MSE) between the predicted and observed spatial gene expressions. The CSS provides a measure of the correlation between predicted and observed gene expressions, whereas the MSE serves as an indicator of absolute gene value recovery ability. Consequently, a lower MSE value and a higher CSS value signify superior prediction performance. We tracked these metrics at the cell level to assess the preservation of crucial cellular characteristics, and at the gene level to measure the correlation between predicted and measured spatial profiles.
    
    We observe that STORM is consistently among the top performing methods in each metric, and in terms of the cell-wise cosine similarity it ranks as the best. Its performance is otherwise minorly edged out by the popular stImpute method. However, we emphasize that stImpute relies on gene relations inferred by latent space NLP representations of protein sequences coded for by each gene. This limits the interpretability of the GNN relations. Our STORM gene relations retain considerably more interpretability as they are constructed from experimentally validated PPI interactions along with the learned structure of the data during the tensor factorization. Additionally, the joint latent space representation in stImpute is iteratively trained via backpropagation with the gene imputation loss as its objective. The STORM joint latent space is learned solely from the biologically motivated joint subspace structure of the data, with no input from the gene imputation task, yet still manages to attain competitive performance and even outperform stImpute at the cell level. Ultimately, STORM achieved average scores of 0.683 and 0.735 with respect to cell-wise MSE and CSS and 0.649 and 0.635 with respect to gene wise metrics. Other techniques, such as Tangram, tend to perform very poorly with respect to the MSE metric. This may be due to its tendency to over predict the expression intensity for individual genes. The lower performance of Tangram could be because it is a tool for analyzing various aspects of spatial transcriptomics data but is not specifically optimized for gene imputation.
    
    \subsection{STORM Tracks Temporal Progression Trends in Developing StereoSeq Mouse Embryos}

    Finally, we move to utilizing STORM to examine the developmental stages of three mouse embryo slices  sampled at time stages of E10.5, E11.5, and E12.5 via the StereoSeq platform \cite{Chen2022StereoSeq}. In this manner we are able to investigate the spatiotemporal heterogeneity of tissue structures during mouse organogenesis. The original data suffered from batch effects, yet we can observe in the STORM embedding an effective integration which produces spatial domains that demonstrate consistent localization. Leiden was used to cluster the spatial regions of each tissue and the computed clusters were then annotated based on the presence of differentially expressed genes. 
    
    We note that in each layer, the predicted clusters show a high concordance with known marker genes of the major organs. For example, we note the detection of the forebrain region marked by the Hes5 gene, the Hindbrain marked by Ina, the hypothalamus by Six6, the cartilage structure by Col2a1, the heart by Myl7, the olfactory bulb and edge of the forebrain by Tbr1, and the epidermis by Krt5. Additionally, we can delineate the gastrointestinal tract and liver regions via the spatiotemporal signal of Cpox. Cpox encodes coproporphyrinogen oxidase, a mitochondrial enzyme in the heme biosynthesis pathway, and CPOX/heme-pathway activity is induced during erythroid differentiation \cite{Taketani2001CPOX_K562}. At E11.5, Cpox shows an elevated signal in the hepatic region and along the developing gut tube, which reflects increased heme/erythroid activity in associated vasculature and surrounding mesenchyme. By E12.5, the signal becomes more concentrated in the liver, consistent with the fetal liver becoming the major site of hematopoiesis around E11.5-E12.5 \cite{Crawford2010HepatobiliaryAtlas}.
 Ultimately, all the major tissues and organs (e.g., heart, liver and brain) were detected and each generally was consistent with the well-known anatomical structure of the organ. The primary exception to this is the detection of the heart region in the E10.5 stage, which was not as spatially localized as in the later stages. 
    
    \begin{figure}[H]
        \centering
        \includegraphics[width=\linewidth]{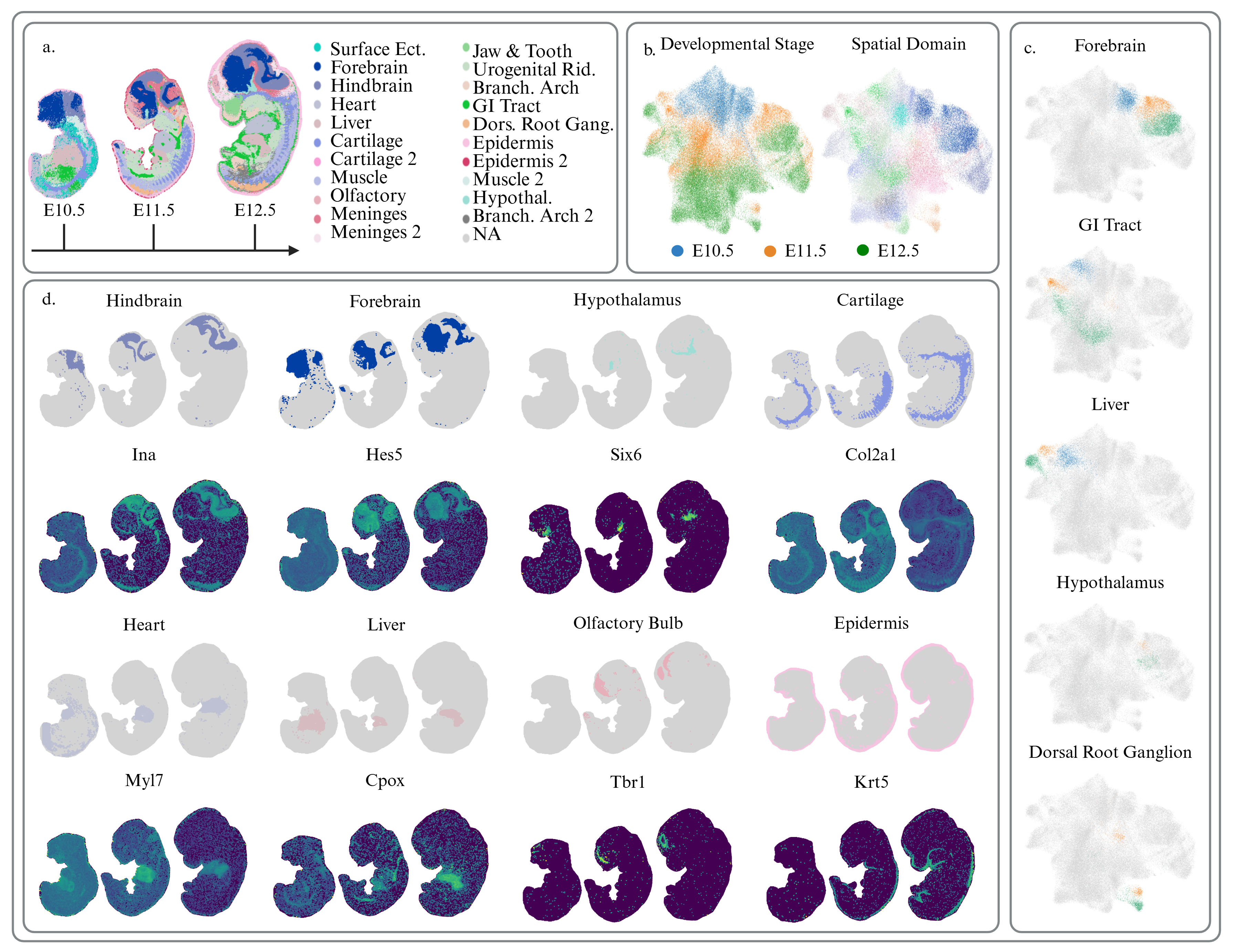}
        \caption{STORM tracking of mouse embryo evolution.   a. The STORM predicted spatial domains of the mouse embryo at developmental stages E10.5, E11.5, and E12.5. Clusters are annotated via the presence of known marker genes and show clear correspondence with known anatomical structure of the developing mouse embryo. b. UMAP embeddings produced by STORM colored by spatial domain as well as developmental stage label. These plots combine to demonstrate that STORM is able to integrate the samples without blurring the temporal differences between them. c. UMAP plots for several select spatial domains show that STORM embeddings depict clear developmental trends between time steps. Clusters in each development stage progress sequentially in the embedding space. Their co-localization implies the amelioration of batch effects, while their sequential ordering reflects that genuine biological variation has not been blurred. d. Plots of select spatial domains as well as differentially expressed marker genes within these domains. This demonstrates that the STORM domains show high concordance with known marker genes in the developing mouse embryo. }
        \label{fig:dynamic}
    \end{figure}

    Notably, we observe in Figure \ref{fig:dynamic}c that these spatial domains between each developmental stage also demonstrate reasonably co-localization patterns based on the UMAP visualization. Clusters in the embedding space generally form a compact manifold with a per-stage structure extending from E10.5 to E11.5 to E12.5 sequentially, reflecting temporal progression and development. This highlights the ability of STORM to integrate this temporal data in such a way as to alleviate batch effects while simultaneously preserving the stage specific variation. The stages overlap broadly rather than forming three completely disjoint clouds, while domains still show structured, directional shifts rather than randomly over-mixing. 
    
    \section{Discussion}
    Spatial transcriptomics is a powerful experimental approach that measures gene expression in its original spatial context. The applications of this information are manifold. However, it is known that there exist critical limitations in multi-slice spatial transcriptomics integrated analysis. Additionally, there is generally a recognized trade-off in single-cell resolution with a limited gene targeting depth. Here, we present STORM, a biologically interpretable tensor based framework for overcoming these limitations by jointly embedding multiple experiments into a shared latent space where the issue of batch effects has been ameliorated, but genuine biological variation between samples has been appropriately preserved. This is accomplished via  a Cartesian graph regularized irregular tensor decomposition on a set of gene expression matrices (spatially resolved or not). We begin by finding low dimensional subspaces of the original cell space and rotating them in this low dimensional space to ameliorate batch effects. Latent space features across slices are coupled and their relative importance per slice is scaled, and the gene expression activity is then mapped into these spatial subspaces. The vectorized tensor reconstruction is then smoothed across the Cartesian product / block structure of spatial, cross sample, and PPI graphs. This simultaneously encourages spatially adjacent spots within a sample, as well as spots paired between samples, to attain similar latent representations, while functionally related genes are mapped similarly into the spatial subspace. Ultimately, the resultant shared latent space structure of the datasets lends itself to a variety of powerful downstream analyses. 

    When utilizing STORM to horizontally integrate the anterior and posterior sections of a mouse brain sample, we observe that, relative to other popular approaches, our model displayed a superior ability to accurately segment the major regions of the brain in concordance with the Allen Mouse Brain Atlas. Specifically, the model was able to accurately group cross-sample domains while managing not to over-smooth to the extent that real variation between the sections became obscured. Additionally, when performing vertical integration on the Visium DLPFC data, STORM achieves state-of-the-art performance relative to other popular methods both in terms of spatial domain detection accuracy and joint embedding quality. Notably, STORM achieved this feat without any reliance on deep learning frameworks, which most of the other models employed. When applied to cell type deconvolution, STORM universally outperforms the results of a recent deconvolution benchmarking study. When tested on synthetic low resolution spatial data against the ground truth single-cell resolution, STORM generally outperformed the other methods When applied to real case studies of mapping single cells to spatial locations in the human lymph node and DLPFC, STORM was found to map cells in a manner which showed good correspondence with known cellular compartments of the lymph node, as well as layer specific distributions of cells in the DLPFC data. By utilizing STORM paired with a GNN framework for gene imputation, we again observed state-of-the-art or near state-of-the-art results on all tested data. Finally, when used to track temporal developments in mouse organogenesis, STORM revealed meaningful progression trends across samples and was able to accurately identify the major tissue and organ regions at each stage of development, showing a high concordance with marker gene distributions. Ultimately, we have demonstrated in this study the versatility and impressive performance of our STORM on a large variety of foundational computational problems in the field of spatial transcriptomics. 

    The spatial subspace framework in STORM comes out of a natural biological motivation. We recognize the assumption that gene expression variation in tissues is spatially structured. At the spatial domain level, cells belong to sets of recurring microenvironmental contexts, and so cells can be viewed as different samples of the same underlying niche-level expression structures. This is commonly exploited in graph-based computational methods which aggregate information via graph message passing. In our case, by projecting cells into spatial subspaces via the tensor of frames, we choose to emphasize the consistent domain-scale structure to ameliorate the gene expression variability between samples. This resolves the irregularity in the different number of cells between slices by focusing on a uniform number of spatial structures. Additionally, this addresses the current conceptual weakness in single slice spatial transcriptomics analysis which explores the structure of microenvironments as isolated 2D structures, when in reality they act as coordinated 3D environments. 

    Tensor based computational approaches to single cell omics data analysis have been previously explored in some contexts. In spatial transcriptomics, this has particularly been in the area of tensor completion for gene denoising. Song et al. performed a graph-guided neural tensor decomposition utilizing a three layer neural network to learn non-linear relations between the two dimensional spatial modes and gene modes. The tensor reconstruction was then assumed to contain denoised gene expression trends \cite{Song2023GNTD}. In the field of multi-omics integration Braytee et al. performed a 3 step  autoencoder, tensor construction, CDP decomposition framework for extracting interpretable latent factors and performing survival analysis \cite{Braytee2024CancerRiskMultiOmics}. Classical methods such as MONTI utilized Non-negative Tensor Decomposition for interpretable alignment of samples across modalities \cite{Jung2021MONTI}. To recover spatio-temporal information in spatial transcriptomics data, Zhou et al. proposed learning a four-dimensional transition tensor and spatial-constrained random walk, thereby reconstructing cell-state-specific dynamics and spatial state transitions via both short-time local tensor streamlines between cells and long-time transition paths among attractors \cite{Zhou2024STT}. However, no such studies make use of the cross-sample graph regularized irregular tensor decomposition framework for the joint embedding of spatial transcriptomics experiments. Such an approach enables a meaningful analysis of data across conditions, as well as across modalities, such as jointly analyzing sc-RNA-seq and ST data. Naturally, this framework could extend to other such types of multi-modal or multi-omics integration tasks, and can continue to see superior performance in these areas in future work. 
   \section{Methods: STORM Model}
\label{sec:method}
We consider $n_b$  slices/batches/samples of gene expression matrices. Slice $k$   contains $n_{s_k}$ spatial spots and $n_{g_k}$ genes, with expression matrix $X_k\in\mathbb{R}^{n_{s_k}\times n_{g_k}}$. After intersecting gene sets across slices, we obtain a common gene set of size $n_g$, and hence $X_k\in\mathbb{R}^{n_{s_k}\times n_{g}}$. The multi-slice spatial transcriptomics setting is inherently multi-way: gene expression varies jointly across spatial location, gene identity, and experimental condition. Modeling these interactions through matrix concatenation collapses the slice structure and obscures mode-specific geometry. Tensor decompositions, such as CP/PARAFAC \cite{harshman1970foundations} %, Tucker decomposition \cite{tucker1966some,cai2021mode},
and PARAFAC2 \cite{kiers1999parafac2}, provide principled multilinear frameworks for modeling such multi-way data. In particular, PARAFAC2 is designed to accommodate settings in which one mode varies in dimension across slices while preserving a shared latent covariance structure.  %For a comprehensive overview of tensor decompositions, see \cite{kolda2009tensor}. 

To preserve the multi-way structure of spatial transcriptomics data while accommodating slice-specific spatial dimensions, we represent the collection  $\{X_k\}_{k=1}^{n_b}$ as an irregular third-order tensor  $\mathcal{X} \in  \mathbb{R}^{n_{s_k} \times n_g \times n_b}$, whose modes correspond to spatial locations, genes, and slices.   The spatial mode is irregular as $n_{s_k}$ may differ across slices. STORM constructs a unified latent representation by aligning each slice through a low-dimensional spatial subspace, thereby resolving spatial irregularity by aligning the samples at the level of their biologically meaningful microenvironment structures (see  Figure~\ref{fig:STORM}). The shared latent space can then be mapped back to the original spatial resolution for downstream analysis.

\subsection{Spatial Subspace Alignment}
To address irregularity in the number of spatial spots, each slice $X_k$ is modeled through an $R$-dimensional spatial subspace, which can be naturally interpreted as the microenvironment structures of the tissue, with orthonormal basis
\[
Q_k \in \mathbb{R}^{n_{s_k} \times R},
\qquad
Q_k^\top Q_k = I_R.
\]
Collecting these slice-specific frames yields
\[
\mathcal{Q} = \{Q_k\}_{k=1}^{n_b}.
\]
The orthonormality constraint ensures that each $Q_k$ defines a rotation into a shared $R$-dimensional microenvironment space while preserving local geometry. Projecting into these subspaces yields a regular spatial subspace tensor
\[
\mathcal{W} \in \mathbb{R}^{R \times n_g \times n_b}.
\]

\subsection{PARAFAC2 Factorization}
\begin{Def}[PARAFAC2 Model]
Given matrices $\{X_k\}_{k=1}^{n_b}$ with 
$X_k \in \mathbb{R}^{n_{s_k}\times n_g}$, 
a rank-$R$ PARAFAC2 factorization consists of matrices
$U_k \in \mathbb{R}^{n_{s_k}\times R}$,
$S_k \in \mathbb{R}^{R\times R}$,
and $V \in \mathbb{R}^{n_g\times R}$
such that
\[
X_k \approx U_k S_k V^\top,
\]
and the Gram matrices of the slice-specific factors satisfy
\[
U_k^\top U_k = \Phi
\quad \text{for all } k,
\]
for some fixed positive semidefinite matrix 
$\Phi \in \mathbb{R}^{R\times R}$.
\end{Def}
STORM factorizes the subspace tensor as
\[
\mathcal{W} = \mathcal{D} \times_1 H \times_2 B,
\]
where
\begin{itemize}
\item $H \in \mathbb{R}^{R \times R}$ captures shared latent covariance structure across slices,
\item $B \in \mathbb{R}^{n_g \times R}$ contains shared gene loadings,
\item $\mathcal{D}_{::k} = D_k \in \mathbb{R}^{R \times R}$ is diagonal and modulates slice-specific component strengths.
\end{itemize}
Accordingly, each frontal slice satisfies
\[
W_k = H D_k B^\top.
\]
Mapping back to the original spatial coordinates yields
\[
X_k \approx Q_k H D_k B^\top.
\]
\begin{thm}
\label{thm:shared_cov}
Assume the STORM parameterization
\[
X_k \approx Q_k H D_k B^\top,\qquad k=1,\dots,n_b,
\]
satisfies the following conditions:
\begin{enumerate}[label=(\roman*)]
  \item For each slice $k$, the spatial frame $Q_k\in\mathbb{R}^{n_{s_k}\times R}$ has orthonormal columns, $Q_k^\top Q_k = I_R$.
  \item $H\in\mathbb{R}^{R\times R}$   has full column rank i.e., $H$ is invertible.
  \item Each $D_k\in\mathbb{R}^{R\times R}$ is diagonal.
  \item $B\in\mathbb{R}^{n_g\times R}$ has full column rank.
\end{enumerate}
Define the slice-specific embedded factors
\[
U_k := Q_k H \in\mathbb{R}^{n_{s_k}\times R}.
\]
Then for every slice $k$,
\[
U_k^\top U_k = H^\top H,
\]
i.e. the Gram  matrix of the embedded factors $U_k$ is \emph{identical} across slices. Consequently, all slices share the same latent covariance geometry in the aligned $R$-dimensional subspace; slice-specific variability in component \emph{magnitudes}  is captured by the diagonal scaling matrices $\{D_k\}$.
\end{thm}
\begin{proof}
Direct calculation using $Q_k^\top Q_k = I_R$ gives
\[
U_k^\top U_k = (Q_k H)^\top (Q_k H) = H^\top (Q_k^\top Q_k) H = H^\top H,
\]
which is independent of $k$, proving the claim.
\end{proof}
\noindent Theorem~\ref{thm:shared_cov} establishes that the embedded slice-specific factors share an identical Gram matrix. This  implies that the STORM factorization satisfies the defining PARAFAC2 constraint, as stated below.
\begin{cor}
Under the assumptions of Theorem~\ref{thm:shared_cov}, 
define $U_k = Q_k H$, $S_k = D_k$, and $V = B$. 
Then the STORM factorization satisfies
\[
X_k \approx U_k S_k V^\top,
\]
and the Gram constraint
\[
U_k^\top U_k \equiv H^\top H
\]
holds for all $k$. 
Hence STORM belongs to the PARAFAC2 model class.
\end{cor}

As the STORM core factorization is a PARAFAC2 model, it inherits the standard PARAFAC2 indeterminacies, namely permutation and scaling of components (and, in certain cases, sign indeterminacy). PARAFAC2 solutions are known to be essentially unique under generic non-degeneracy conditions and sufficient diversity across slices, with formal uniqueness analyses provided in the PARAFAC2 literature (e.g., \cite{ten1996some}) and discussed in foundational algorithmic treatments \cite{kiers1999parafac2}. Consequently, under mild conditions, the shared gene loadings $B$ and slice-specific component weights $\{D_k\}$ are identifiable up to permutation and scaling.
%This structure ensures that latent microenvironment directions share a common covariance geometry across slices, while slice-specific scaling matrices $D_k$ capture condition-dependent variation.

\subsection{Graph-Regularized Reconstruction}

To incorporate spatial and functional priors into the tensor factorization, we introduce graph regularization jointly across spatial and gene modes. Let $n_s = \sum_{k=1}^{n_b} n_{s_k}$. For any tensor $\mathcal{Y} = \{Y_k\}_{k=1}^{n_b}$ with frontal slices $Y_k \in \mathbb{R}^{n_{s_k} \times n_g}$, define the stacking operator
\begin{equation}
\mathrm{mat}(\mathcal{Y})
=
\begin{bmatrix}
Y_1 \\
\vdots \\
Y_{n_b}
\end{bmatrix}
\in \mathbb{R}^{n_s \times n_g}.
\end{equation}
Functional relationships between genes are modeled using a graph with adjacency matrix 
$A_g \in \mathbb{R}^{n_g \times n_g}$ (e.g., a protein–protein interaction network). 
Let
\[
D_g = \mathrm{Diag}(A_g \mathbf{1}),
\qquad
L_g = D_g - A_g
\]
denote the corresponding degree matrix and Laplacian of the gene-gene interactions.

 It is further generally assumed that gene expression varies smoothly in spatial neighborhoods of the tissue. Within each slice, spatial adjacency is constructed using physical coordinates $p_{k,i} \in \mathbb{R}^2$ (e.g., radius graph or $k$-nearest neighbors (k-NN)). 
To smooth variation across slices, one must couple spots across the slices. In this work, cross-slice edges are defined via shared embeddings or optimal transport alignments. We therefore build a block spatial Laplacian
\begin{equation}\label{eqn:spatialLap}
L_s=
\begin{bmatrix}
L_s^{(11)} & \cdots & L_s^{(1n_b)}\\
\vdots & \ddots & \vdots\\
L_s^{(n_b1)} & \cdots & L_s^{(n_bn_b)}
\end{bmatrix},
\qquad
L_s^{(k\ell)}\in\mathbb{R}^{n_{s_k}\times n_{s_\ell}}.
\tag{5}
\end{equation} 
where diagonal blocks $L_s^{(kk)}$ encode within-slice adjacency (threshold $\varepsilon_1$), while off-diagonal blocks
$L_s^{(k\ell)}$ encode between-slice adjacency (threshold $\varepsilon_2$).   

For slice $k$, let $p_{k,i}\in\mathbb{R}^2$ be the physical coordinate of spot $i$.
We build a within-slice adjacency $A_s^{(kk)}\in\mathbb{R}^{n_{s_k}\times n_{s_k}}$ using a radius or $k$-NN rule, e.g.,
\[
(A_s^{(kk)})_{ij}
=
\mathbf{1}\{\|p_{k,i}-p_{k,j}\|_2\le \varepsilon_1\}.
\]

Let $z_{k,i}\in\mathbb{R}^R$ denote the aligned latent coordinate of spot $i$ in slice $k$. This can be obtained via an optimal transport based coupling such as PASTE, or more simply by a naive joint embedding into some $R$ dimensional shared space, say by PCA. 
For $k\neq \ell$, we define cross-slice adjacency blocks $A_s^{(k\ell)}\in\mathbb{R}^{n_{s_k}\times n_{s_\ell}}$ by
\[
(A_s^{(k\ell)})_{ij}
=
\mathbf{1}\{\|z_{k,i}-z_{\ell,j}\|_2\le \varepsilon_2\},
\]
and enforce an undirected graph by setting $A_s^{(\ell k)}=(A_s^{(k\ell)})^\top$.

Let $A_s$ be the block matrix with blocks $\{A_s^{(k\ell)}\}_{k,\ell}$ and define
$D_s=\mathrm{Diag}(A_s\mathbf{1})$ and $L_s=D_s-A_s$.
Then for $k\neq \ell$, $L_s^{(k\ell)}=-A_s^{(k\ell)}$ (rectangular in general),
and $L_s$ is symmetric with $L_s^{(k\ell)}=(L_s^{(\ell k)})^\top$. This defines a spatial block graph Laplacian successfully coupling all cells within each slice as well as between slices to further smooth batch effects. 

The goal of graph regularization is to encourage smooth variation of reconstructed gene expression across spatial neighborhoods and functionally related genes. 
Smoothness is imposed through the Cartesian product Laplacian. 
For any reconstruction tensor $\mathcal{Y}$, we define the regularization  as
\begin{align}
\mathcal{R}(\mathcal{Y})
&=
\mathrm{Tr}\!\big(\mathrm{mat}(\mathcal{Y})^\top L_s\,\mathrm{mat}(\mathcal{Y})\big)
+
\mathrm{Tr}\!\big(\mathrm{mat}(\mathcal{Y}) L_g\,\mathrm{mat}(\mathcal{Y})^\top\big) \nonumber\\
&=
\mathrm{vec}(\mathrm{mat}(\mathcal{Y}))^\top
(L_g \oplus L_s)
\mathrm{vec}(\mathrm{mat}(\mathcal{Y})),
\end{align}
where
\[
L_g \oplus L_s
=
(L_g \otimes I_{n_s})
+
(I_{n_g} \otimes L_s)
\]
is the Cartesian graph product Laplacian.

Graph-regularized tensor factorizations have been widely studied to incorporate relational or manifold structure into multilinear models, including Laplacian-regularized nonnegative tensor factorization and graph-constrained Tucker  decompositions \cite{chen2022unsupervised,qiu2020generalized}, as well as coupled graph–tensor factorization frameworks \cite{ioannidis2019coupled}. These approaches typically impose smoothness on shared tensor modes. In contrast, STORM integrates graph regularization directly within a PARAFAC2 framework, enabling joint smoothing across spatial and gene modes while preserving irregular slice-specific spatial dimensions.

Recall that the reconstruction tensor is defined as
\[
\mathcal{Y} = \mathcal{Q} \Delta \mathcal{W},
\qquad
(\mathcal{Q} \Delta \mathcal{W})_{::k} = Q_k H D_k B^\top.
\]
The graph-regularized PARAFAC2 model is formulated as
\begin{equation}
\label{eq:graph_tensor_model}
\min_{\mathcal{Q},\,\mathcal{W}}
\;
\|\mathcal{X} - \mathcal{Q} \Delta \mathcal{W}\|_F^2
+
\gamma\,\mathcal{R}(\mathcal{Y})
\quad
\text{s.t.}\quad
Q_k^\top Q_k = I_R \;\; \forall k \text{ and } \mathcal{Y} = \mathcal{Q} \Delta \mathcal{W}
\end{equation}
Because the reconstruction depends left-multiplicatively on the slice-specific spatial frames $\{Q_k\}$, the graph penalty acts directly on the spatial subspaces while preserving the PARAFAC2 Gram structure. The resulting objective is nonconvex due to bilinear coupling and orthogonality constraints. We solve it using an ADMM-based alternating minimization scheme (Supplementary Information Section~4). Adapting AO-ADMM to solve STORM has the following benefits. First, AO-ADMM is more general than other methods in the sense that the loss function doesn’t need to be differentiable. Additionally, it is simple to implement, parallelize, and it is easy to incorporate a wide variety of constraints that can be obtained using simple element-wise operations- leaving the door open for additional biologically informed regularization. There are also computational savings gained by using the Cholesky decompositions in each of the solver steps, while the splitting scheme can be applied to large scale data. Data can be distributed across different machines and optimized locally with communication on the primal, auxiliary and dual variables between the machines. As spatial transcriptomics experiments continue to grow in their spatial resolution and gene targeting depth this scalability becomes increasingly important. %Under standard regularity and boundedness assumptions, nonconvex ADMM-type algorithms are known to converge to first-order stationary points (e.g., \cite{hong2016convergence}). In practice, we observe stable convergence across all datasets. 

The resulting embeddings $\{Q_k H D_k\}_{k=1}^{n_b}$ provide vertically and horizontally integrated latent coordinates for each slice, which are subsequently used for downstream analyses including cell-type deconvolution and gene imputation. Detailed procedures are provided in Supplementary Information Section~4, and parameter ablation studies assessing performance sensitivity and computational efficiency are reported in Supplementary Information Section~5.

 \section{Acknowledgements}
This work was supported in part by NIH grant R35GM148196, National Science Foundation grants DMS2052983 and DGE2152014,  Michigan State University Research Foundation, and  Bristol-Myers Squibb 65109.

 \section{Data and Code Availability}
 The Visium DLPFC data is available in h5ad format at \href{https://weilab.math.msu.edu/DataLibrary/SpatialTranscriptomics/}{WeiLab}. The anterior and posterior mouse brain sections can be found at \href{https://zenodo.org/records/10698931}{Zenodo}. The MERFISH data is available via the SquidPy python library. The seqFISH+ data can be obtained via the scVI python library. The Human Lymph Node ST sample and sc-RNA-seq reference file are publicly available for download through S3 bucket. The DLPFC sc-RNA-seq reference file is available at \href{https://www.ncbi.nlm.nih.gov/geo/query/acc.cgi?acc=GSE144136}{GSE144136}. The ST and sc-RNA-seq data used in the gene imputation analyses are each available at the following repository \href{https://doi.org/10.5281/zenodo.3967290}{Zenodo}. The StereoSeq mouse organogenesis data is available at \href{https://en.stomics.tech/resources/stomics-datasets/list.html}{STOmics}. 
 %A python implementation of STORM as well as files for reproducing the results stated in this study can be found under the following GitHub Repository \href{https://github.com/seanfcottrell/STORM/tree/main}{https://github.com/seanfcottrell/STORM}.
 \newpage
%\bibliographystyle{unsrt}
% \bibliographystyle{abbrv}
%\myexternaldocument{supplementary}

 \bibliographystyle{plain}
\bibliography{refs}

\end{document}